\RequirePackage{fix-cm}
\documentclass[smallextended]{svjour3}       
\smartqed  
\usepackage{graphicx}
\usepackage{amsfonts}
\usepackage{amsmath}
\usepackage{amssymb}
\usepackage{mathtools}
\usepackage{hyperref}
\usepackage{braket}

\newcommand{\tr}{\text{tr}}

\begin{document}

\title{Number of quantum measurement outcomes as a resource
}


\author{Weixu Shi         \and
	Chaojing Tang 
}


\institute{Weixu Shi \and Chaojing Tang \at
	Department of Electronic Science, National University of Defense Technology, Deyaroad 109, 410073 Changsha, China \\
	\email{wx.shi@nudt.edu.cn}            
}

\date{Received: date / Accepted: date}

\maketitle
\begin{abstract}
	Recently there have been fruitful results on resource theories of quantum measurements. Here we investigate the number of measurement outcomes as a kind of resource. We cast the robustness of the resource as a semi-definite positive program. Its dual problem confirms that if a measurement cannot be simulated by a set of smaller number of outcomes, there exists a state discrimination task where it can outperforms the whole latter set. An upper bound of this advantage that can be saturated under certain condition is derived. We also show that the possible tasks to reveal the advantage are not restricted to state discrimination and can be more general.
\end{abstract}

\section{Introduction}
\label{intro}
Quantum measurements are of central interest in quantum information theory. They collapse the quantum state to the classical world which we can perceive, thus allowing us to probe the quantum system. One fundamental property of a measurement is the number of outcomes, or more specifically, the effective one. We focus on the effective number of outcomes because one may simulate a measurement from ones with less outcomes by mixing, relabeling and post-selecting \cite{oszmaniec_simulating_2017}. However, such classical processing like these does not increase the effective number of outcomes.  In the following, when we speak of the number of outcomes of a measurement, we refer to the effective one.

One can observe the difference between the effective number of outcomes and the apparent one at the level of behaviors. It has been proved that quantum theory cannot be explained by n-chotomic theories, where all measurements are constructed from n-outcome ones \cite{kleinmann_quantum_2016,kleinmann_proposed_2017,hu_observation_2018}. Specifically, the importance of the number of outcomes is well demonstrated by observing  the existence of quantum correlations which are not in the nonsignaling set produced by measurements with smaller number of outcomes \cite{kleinmann_quantum_2016}. Besides, it is shown that three-outcome measurements can reveal more nonclassicality in some quantum systems than binary ones \cite{nguyen_quantum_2020}. Other than device-independent  semi-device-independent certification, based on the assumption of upper bounded distinguishability, some three-outcome measurements can be certified from the those of binary measurements \cite{shi_semi-device-independent_2019}. 

Since it enables a larger set of quantum behavior, a measurement with more outcomes can be advantageous in quantum information processing tasks. It is natural to see this property as a kind of resource. Resource theory provides a framework to quantify the resource, look into the operational meaning and explore the possibility of application. In resource theory, one can define a set of free states and operations that cannot increase the resource are free operations.  Robustness, as a measure of how ``far''  is a state away from the convex set of free states, has been proved to be an appropriate quantifier of the outperformance. The recent years have seen fruitful results on quantum resources \cite{chitambar_quantum_2019}, among which to our most interest are resource theories regarding quantum measurements. Incompatibility of measurements was first focused on and was associated operationally with state discrimination tasks \cite{uola_quantifying_2019,skrzypczyk_all_2019,buscemi_complete_2020,carmeli_quantum_2019}. More generally, for every resourceful measurement, there exists a state discrimination tasks in which it outperforms all the non-resourceful ones \cite{uola_quantifying_2019,oszmaniec_operational_2019,takagi_general_2019}. 

Here we investigate the number of measurement outcomes as a resource. We show how to compute the robustness of a given measurement via semi-definite positive programming (SDP). The dual problem confirms the advantage of larger number of outcomes in state discrimination tasks. Next we give an upper bound of the advantage, which can be saturated under certain condition. Lastly, we prove that a more general kind of prepare-and-measure experiment can reveal the advantage of resourceful measurements.

\section{Advantage quantified by robustness}
\label{sec:robustness}
First we introduce the resource of number of outcomes.
The free operations are defined as all the operations that cannot increase the effective number of outcomes, including mixing and relabeling. Let $P_m$ be set of all $m$-outcome positive operator valued measurements (POVMs). The free set $F^n_m$, as the set of all the $m$-outcome POVMs that can be simulated by $n$-outcome ones, is a convex and closed subset of $P_m$.  The POVM $\mathcal{O}$ in $F_m^n$ can be written as
$O_b=\sum_{a,x}p(x)p(b|a,x)Q_{a|x}$, where $\{Q_{a|x}\}$ is an ensemble of $n$-outcome measurements, with $x$ labeling the measurements and $a$ the outcomes, and $p(b|a,x)$ denotes the probabilistic strategy to mix the measurements and relabel the outcomes.

Without loss of generality, we can restrict $p(b|a,x)$ to deterministic function. Rewrite $x$ as the combinations of $m$ outcomes taken $n$ at a time, namely, $x=x_1 x_2 \dots x_n$, with $x_i=1,\dots,m$ and $x_i<x_j$ if $i<j$. Then the deterministic relabeling is $p(b|a,x)=D(b|a,x)=\delta_{b,x_a}$.

The robustness of a POVM, $\mathcal{M}=\{M_b\}$, with respect to $F^n_m$ (in the following we use $F$ for simplicity), is defined as 
\begin{align}\label{eq:def_rob}
	R_{F}(\mathcal{M})=\{t\ge0\left|\frac{M_b+tN_b}{1+t}=O_b\in F\right.\},
\end{align}
where $\mathcal{N}=\{N_b\}\in P_m$. It characterizes the relative distance from the measurement to the surface of the free set. Another interpretation of it is the mininum ``noise'' that can make the measurement not resourceful (fall into the free set). The robustness can be cast as an SDP:
\begin{align}
	\min \quad & 1+R_F(\mathcal{M})=\frac{1}{d}\sum_{a,x,b} D(b|a,x)\tr(\tilde{Q}_{a|x}) \nonumber\\
	\text{s.t.} \quad & \tilde{Q}_{a|x}\ge 0 \quad\forall a,x,\nonumber\\
	& \sum_{a,x}D(b|a,x)\tilde{Q}_{a|x}-M_b\ge 0 \quad\forall b, \nonumber\\
	& \sum_a \tilde{Q}_{a|x}-\frac{1}{d}\sum_a \tr(\tilde{Q}_{a|x})\mathbb{I}=0 \quad\forall x,
\end{align}
where $\tilde{Q}_{a|x} = (1+t)p(x)Q_{a|x}$ and $N_b$ have been substituted.
The dual program reads
\begin{align}
	\max \quad & \sum_b \tr(M_b Y_b) \nonumber\\
	\text{s.t.} \quad & Z_x-\frac{1}{d}\tr(Z_x)\mathbb{I}+\sum_b D(b|a,x)Y_b\leq \frac{1}{d}\sum_b D(b|a,x)\mathbb{I} \quad \forall a,x, \nonumber\\
	&Z_x \text{ is Hermitian} \quad \forall x, \nonumber\\
	&Y_b \geq 0 \quad \forall b.
\end{align}
Strong duality holds because $\tilde{Q}_{a|x}=\mathbb{I}$ is a strictly feasible point for the primal problem. 

By choosing the state ensemble as $\mathcal{E}=\{Y_b/\tr(\sum_b Y_b)\}$ and combining the positivity inequality derived from Eq.~\ref{eq:def_rob}, one can see that robustness implicates the maximal advantage of the measurement in state discrimination tasks \cite{uola_quantifying_2019,oszmaniec_operational_2019}, namely,
\begin{align} \label{eq:max_adv}
	\max_\mathcal{E} \frac{P_\text{guess}(\mathcal{E},\mathcal{M})}
	{\max_{O_x\in F_n} P_\text{guess}(\mathcal{E},\mathcal{O})}=1+R_F(\mathcal{M}),
\end{align}
among which $P_\text{guess}(\mathcal{E},\mathcal{N})=\sum_b\tr(\tilde{\rho}_b N_b)$ is the probability of guessing correctly in a state discrimination game of instance $\mathcal{E}=\{\tilde{\rho}_b\}$, which has absorbed the prior probabilities into the density operators.
This proves that for a resourceful $m$-outcome measurement, there exists a state discrimination task in which the measurement outperforms all the non-resourceful ones. 

Note that Eq.~\ref{eq:max_adv} also provides a semi-device-independent certification of the number of outcomes larger than $n$. On the perfect knowledge of the state preparation, one can compute the denominator of Eq.~\ref{eq:max_adv}. If the guessing probability in the experiment goes beyond that, the existence of measurements with more than $n$ outcomes can be concluded.

\section{Maximal advantage}
Last section has shown that every non $n$-outcome simulable measurement can show advantage over all the simulable ones in at least one state discrimination task. One may next wonder what is its greatest advantage.
Via see-saw method on the dual problem, one can have an empirical result as well as the corresponding quantum realization. Here we give an analytical upper bound of the maximal advantage.
\begin{proposition}\label{thm:existpmtask}
	The maximal advantage of $P_m$ over $F^m_n$ is upper bounded as
	\begin{align}
		\max_{\mathcal{M}\in P_m} 1+R_F(\mathcal{M})\leq \frac{m}{n}.
	\end{align}
	The inequality is saturated when $d\ge m$.
\end{proposition}

\begin{proof}
	By substituting $O_b=\sum_{a,x}D(b|a,x)p(x)O_{a|x}$ into the denominator of Eq.~\ref{eq:max_adv}, we have 
	\begin{align}
		\max_{O_b\in F_n, p(x)} P_\text{guess}(\mathcal{E},\mathcal{O}) &=\max_{O_b\in F_n, p(x)}\sum_b\tr\left(\tilde{\rho}_b\sum_{a,x}D(b|a,x)p(x)O_{a|x}\right) \nonumber \\
		&=\max_{O_b\in F_n, p(x)} \sum_{a,x} p(x)\tr(\tilde{\rho}_{x_a}O_{a|x}) \nonumber\\
		&=\max_{O_b\in F_n, p(x)} \sum_x p(x) q_x P_\text{guess}(\hat{\mathcal{E}}_x,O_x)\nonumber\\
		&=\max_{p(x)} \sum_x p(x)q_x P_\text{guess}(\hat{\mathcal{E}}_x)\nonumber\\
		&\ge \frac{C^{n-1}_{m-1}}{C^n_m}\sum_x q'_x P_\text{guess}(\hat{\mathcal{E}}_x)\nonumber\\
		&=\frac{n}{m}\sum_x q'_x P_\text{guess}(\hat{\mathcal{E}}_x)
	\end{align}
	where $q_x=\sum_{b\in x}\tr(\tilde{\rho}_b)=\sum_a\tr(\tilde{\rho}_{x_a})$, and it leaves $\hat{\mathcal{E}}_x=\{\tilde{\rho}_{x_a}/q_x\}$ a valid ensemble, with the prior probabilities absorbed. The inequality is attained by letting $p(x)=1/|x|=1/C^n_m$ and $C^{n-1}_{m-1}$ is the factor which normalizes $\{p(x)q_x\}$ to a valid probability distribution, $\{q'_x\}$.
	The term $\sum_x q'_x P_\text{guess}(\hat{\mathcal{E}}_x)$ can be interpreted as the score of a state discrimination game with pre-measurement information. To be more specific, the pre-measurement information tells from which $m$-state sub-ensemble the state is chosen. In this case, the measurements can be optimized according to each sub-ensemble. Apparently, this information cannot decrease the probability of guessing correctly over the whole ensemble. Hence it holds that 
	\begin{align}
		\max_{O_b\in F_n, p(x)}P_\text{guess}(\mathcal{E},\mathcal{O})\ge \frac{n}{m}\max_\mathcal{M}P_\text{guess}(\mathcal{E},\mathcal{M}).
	\end{align}
	Note that this inequality holds for any $\mathcal{E}$, we have
	\begin{align}\label{eq:upperbound}
		\max_\mathcal{E}\frac{\max_{\mathcal{M}}P_\text{guess}(\mathcal{E},\mathcal{M})}
		{\max_{O_b\in F_n} P_\text{guess}(\mathcal{E},\mathcal{O})}\leq \frac{m}{n}.
	\end{align}
	When $d\ge m$, letting $\mathcal{E}$ be $m$ uniformly distributed orthogonal states can lead to equality of Eq.~\ref{eq:upperbound}. In this case, optimal $\mathcal{M}$ is the projectors onto the corresponding orthogonal states, optimal $\mathcal{O}$ is the combination of $C^n_m$ measurements with $n$ projectors as elements corresponding to different combinations of choosing $n$ from $m$. \qed
\end{proof}

The above proposition shows that increasing the dimension can no longer augment the largest advantage when the size of the ensemble is larger than the dimension. Besides dimension, larger $m$ can bring larger greatest advantage because $P_m\subseteq P_{m'}$ if $m'> m$. However the increase stops after $m'=d^2$. This is because an extremal measurement on Hilbert space of dimension $d$ can have no more than $d^2$ outcomes \cite{dariano_classical_2005}, which means that any measurements can be simulated by $d^2$-outcome measurements. And since the simulation is a free operation, it does not generate any advantage.

\section{Generalization}
Here we make a generalization regarding tasks that can reveal the advantage of resourceful measurements. We show that not only the state discrimination experiment, but a certain kind of quantum prepare-and-measure experiment can reveal the advantage of the resourceful measurements. Using similar proving skill in Ref.~\cite{oszmaniec_operational_2019}, we have the following theorem.
\begin{proposition}\label{prop:generalized}
	Suppose the score of a quantum prepare-and-measure experiment is given by $
	S=\sum_{x,y,b} c_{x,y,b}p(x)p(b|x,y)
	$,
	where $c_{x,y,b}$ are real coefficients. If $F$ is a convex set of measurement assemblages, and if the linear mapping $f: \{M_{b|y}\}\mapsto\{N_x\}$ characterized by $\{c_{x,y,b}\}$ is a bijection, then for any measurement assemblage $\{M_{b|y}\}\notin F$, there exists an instance of the above experiment where $\{M_{b|y}\}$ strictly outperforms all the members in $F$.
\end{proposition}
\begin{proof}
	Let $\mathcal{M}$ be any resourceful measurement assemblage, i.e. $\mathcal{M}\notin F$. Its image $\mathcal{N}=f(\mathcal{M})$ must be outside $f(F)$. The reason is that if $\mathcal{N}\in f(F)$, it must have an image $\mathcal{M}'$ in $F$, which contradicts with the premise that $f$ is a bijection. Since $f$ as a linear map is convexity preserving, $f(F)$ is a convex set. It follows from the \textit{separating hyperplane theorem} that there exists a hyperplane described by $\{W_i\}$ such that $\sum\tr(W_iN'_i)< 0$ for all $\mathcal{N}'\in f(F)$ and $\sum\tr(W_iN_i)\ge 0$ for all $\mathcal{N} \notin f(F)$. Define $\tilde{W}_i=W_i+|\lambda |\mathbb{I}$, where $\lambda$ is the smallest eigenvalue of $W_i$. We have $\sum\tr(\tilde{W}_xN'_x)<\sum\tr(\tilde{W}_xN_x)$. Letting $\mathcal{E}=\{p(x)\rho_x=\tilde{W}/\sum\tr(\tilde{W}_x)\}$, one can find 
	\begin{align}
		\sum\tr(\tilde{W}_xN_x) = \sum_{x,y,b} c_{x,y,b}p(x) \tr(\rho_x M_{b|y})=S(\mathcal{E},\mathcal{M}),
	\end{align}
	and similarly $\sum\tr(\tilde{W}_xN'_x)=S(\mathcal{E},\mathcal{M'})$.
	We can see that $\mathcal{E}$ gives an instance which reveals the advantage of the resource measurement assemblage.  \qed
\end{proof}

The case $x=(w,a)$ and $c_{x,y,b}=\delta_{b,a}\delta_{w,y}$ corresponds to a state discrimination task with prior-measurement information, which has been so far widely used to discuss the outperformance of resourceful measurements \cite{uola_quantifying_2019,skrzypczyk_all_2019,buscemi_complete_2020,carmeli_quantum_2019}.

The maximal relative advantage of the resource can be related to the robustness of $N_x$ with respect to $f(F)$ as
\begin{align} \label{eq:gen_max_adv}
	\max_\mathcal{E} \frac{S(\mathcal{E},\mathcal{M})}
	{\max_{O_b\in F_n} S(\mathcal{E},\mathcal{O})}=1+R_{f(F)}(\mathcal{N}).
\end{align}

Proposition~\ref{prop:generalized} can give results on advantages if the score has the form $S=\sum_{x,y,b} c_{x,y,b}p(x)p(b|x,y)$. Here we make some discussions on scores defined more generally. 
If $S$ as a function of probabilities is nonlinear, the proof can no longer work because of the use of seperating hyperplane theorem.
But note that $f$ can be generalized to any linear map that has 
$f(A\setminus B)\cap f(B)=\emptyset$ given that $B\subset A$, and non necessary bijections. In other words, the map $f$ preserve the membership relation in their images as originally. It would be interested to find a method to certify this property of $f$ under certain resource.

\section{Conclusion}
In this paper, we investigated number of measurement outcomes as a resource. We characterized the robustness via SDP and derived the dual problem of robustness to confirm the operational meaning of robustness as advantage in state discrimination. We gave un upper bound of the maximal advantage. It would be interesting to have a tighter bound for the case $d<m$. At last, we have shown that a broaden kind of prepare-and-measure experiment can be used to demonstrate the advantage of the resource. The duality just gives us one way to interpret operationally a measure of the resource, more possibilities are to be found on this.


%
\section*{Conflict of interest}
The authors declare that they have no conflict of interest.

\bibliographystyle{spphys}       
\bibliography{reference}   

\begin{thebibliography}{10}
\providecommand{\url}[1]{{#1}}
\providecommand{\urlprefix}{URL }
\expandafter\ifx\csname urlstyle\endcsname\relax
  \providecommand{\doi}[1]{DOI \discretionary{}{}{}#1}\else
  \providecommand{\doi}{DOI \discretionary{}{}{}\begingroup
  \urlstyle{rm}\Url}\fi

\bibitem{oszmaniec_simulating_2017}
M.~Oszmaniec, L.~Guerini, P.~Wittek, A.~Acín, Phys. Rev. Lett.
  \textbf{119}(19), 190501 (2017).
\newblock \doi{10.1103/PhysRevLett.119.190501}.
\newblock
  \urlprefix\url{https://link.aps.org/doi/10.1103/PhysRevLett.119.190501}.
\newblock Publisher: American Physical Society

\bibitem{kleinmann_quantum_2016}
M.~Kleinmann, A.~Cabello, Phys. Rev. Lett. \textbf{117}(15), 150401 (2016).
\newblock \doi{10.1103/PhysRevLett.117.150401}.
\newblock
  \urlprefix\url{https://link.aps.org/doi/10.1103/PhysRevLett.117.150401}

\bibitem{kleinmann_proposed_2017}
M.~Kleinmann, T.~Vértesi, A.~Cabello, Phys. Rev. A \textbf{96}(3), 032104
  (2017).
\newblock \doi{10.1103/PhysRevA.96.032104}.
\newblock \urlprefix\url{http://arxiv.org/abs/1611.05761}.
\newblock ArXiv: 1611.05761

\bibitem{hu_observation_2018}
X.M. Hu, B.H. Liu, Y.~Guo, G.Y. Xiang, Y.F. Huang, C.F. Li, G.C. Guo,
  M.~Kleinmann, T.~Vértesi, A.~Cabello, Phys. Rev. Lett. \textbf{120}(18),
  180402 (2018).
\newblock \doi{10.1103/PhysRevLett.120.180402}.
\newblock \urlprefix\url{http://arxiv.org/abs/1712.06557}.
\newblock ArXiv: 1712.06557

\bibitem{nguyen_quantum_2020}
H.C. Nguyen, O.~Gühne, arXiv:2001.03514 [quant-ph]  (2020).
\newblock \urlprefix\url{http://arxiv.org/abs/2001.03514}.
\newblock ArXiv: 2001.03514

\bibitem{shi_semi-device-independent_2019}
W.~Shi, Y.~Cai, J.B. Brask, H.~Zbinden, N.~Brunner, Phys. Rev. A
  \textbf{100}(4), 042108 (2019).
\newblock \doi{10.1103/PhysRevA.100.042108}.
\newblock \urlprefix\url{https://link.aps.org/doi/10.1103/PhysRevA.100.042108}

\bibitem{chitambar_quantum_2019}
E.~Chitambar, G.~Gour, Rev. Mod. Phys. \textbf{91}(2), 025001 (2019).
\newblock \doi{10.1103/RevModPhys.91.025001}.
\newblock \urlprefix\url{https://link.aps.org/doi/10.1103/RevModPhys.91.025001}

\bibitem{uola_quantifying_2019}
R.~Uola, T.~Kraft, J.~Shang, X.D. Yu, O.~Gühne, Phys. Rev. Lett.
  \textbf{122}(13), 130404 (2019).
\newblock \doi{10.1103/PhysRevLett.122.130404}.
\newblock
  \urlprefix\url{https://link.aps.org/doi/10.1103/PhysRevLett.122.130404}

\bibitem{skrzypczyk_all_2019}
P.~Skrzypczyk, I.~Šupić, D.~Cavalcanti, Phys. Rev. Lett. \textbf{122}(13),
  130403 (2019).
\newblock \doi{10.1103/PhysRevLett.122.130403}.
\newblock
  \urlprefix\url{https://link.aps.org/doi/10.1103/PhysRevLett.122.130403}

\bibitem{buscemi_complete_2020}
F.~Buscemi, E.~Chitambar, W.~Zhou, Phys. Rev. Lett. \textbf{124}(12), 120401
  (2020).
\newblock \doi{10.1103/PhysRevLett.124.120401}.
\newblock
  \urlprefix\url{https://link.aps.org/doi/10.1103/PhysRevLett.124.120401}

\bibitem{carmeli_quantum_2019}
C.~Carmeli, T.~Heinosaari, A.~Toigo, Phys. Rev. Lett. \textbf{122}(13), 130402
  (2019).
\newblock \doi{10.1103/PhysRevLett.122.130402}.
\newblock
  \urlprefix\url{https://link.aps.org/doi/10.1103/PhysRevLett.122.130402}

\bibitem{oszmaniec_operational_2019}
M.~Oszmaniec, T.~Biswas, Quantum \textbf{3}, 133 (2019).
\newblock \doi{10.22331/q-2019-04-26-133}.
\newblock \urlprefix\url{http://arxiv.org/abs/1901.08566}.
\newblock ArXiv: 1901.08566

\bibitem{takagi_general_2019}
R.~Takagi, B.~Regula, Phys. Rev. X \textbf{9}(3), 031053 (2019).
\newblock \doi{10.1103/PhysRevX.9.031053}.
\newblock \urlprefix\url{http://arxiv.org/abs/1901.08127}.
\newblock ArXiv: 1901.08127

\bibitem{dariano_classical_2005}
G.M. D'Ariano, P.L. Presti, P.~Perinotti, J. Phys. A: Math. Gen.
  \textbf{38}(26), 5979 (2005).
\newblock \doi{10.1088/0305-4470/38/26/010}.
\newblock
  \urlprefix\url{https://iopscience.iop.org/article/10.1088/0305-4470/38/26/010}

\end{thebibliography}

%
%

\end{document}